\documentclass{amsart}
\usepackage{amsfonts}
\usepackage{amsmath,amscd,lscape}
\usepackage{amsthm}
\usepackage{amssymb}
\usepackage{latexsym}
\usepackage{mathrsfs}
\usepackage{tikz}
\newcommand{\nc}{\newcommand}
\nc{\ben}{\begin{eqnarray}}
\nc{\een}{\end{eqnarray}}

\newcommand{\beqa}{\begin{eqnarray}}
\newcommand{\eeqa}{\end{eqnarray}}
\nc{\Z}{{\bold Z}}

\usepackage[curve,matrix,arrow,color]{xy}

\usepackage{color}

\newtheorem{prop}{Proposition}[section]

\newtheorem{example}{Example}
\newtheorem{defn}{Definition}[section]

\newtheorem{rem}{Remark}

\newcommand{\cT}{\mathscr{T}}
\newcommand{\cB}{\mathscr{B}}
\newcommand{\cO}{\mathscr{O}}
\newcommand{\CC}{\mathbb{C}}
\newcommand{\ZZ}{\mathbb{Z}}

\newcommand{\be}{\bold{e}}
\newcommand{\bof}{\bold{f}}
\newcommand{\bg}{\bold{g}}
\newcommand{\bh}{\bold{h}}

\numberwithin{equation}{section}
\begin{document}

\title[Higher rank Askey-Wilson algebras]{Higher rank classical analogs of the Askey-Wilson algebra \\ from the $sl_N$ Onsager algebra}
%\dedicatory{}
\author{Pascal Baseilhac$^{*}$}
\address{$^*$ Institut Denis-Poisson CNRS/UMR 7013 - Universit\'e de Tours - Universit\'e d'Orl\'eans
Parc de Grammont, 37200 Tours, 
FRANCE}
\email{pascal.baseilhac@idpoisson.fr}

\author{Nicolas Cramp\'e$^{\dagger}$}
\address{$^\dagger$  Laboratoire Charles Coulomb (L2C), Univ Montpellier, CNRS, Montpellier, FRANCE}

\address{Centre de recherches \\ math\'ematiques,
Universit\'e de Montr\'eal, P.O. Box 6128, Centre-ville Station,
Montr\'eal, H3C 3J7,  CANADA}

\email{crampe1977@gmail.com}

\author{Rodrigo A. Pimenta$^{*,**}$}
\address{$^{**}$ Instituto de Fisica de Sao Carlos, Universidade de Sao Paulo, Caixa Postal 369, 13.560-590, Sao Carlos,
SP, BRAZIL} 
\address{CAPES Foundation, Ministry of Education of
Brazil, Brasilia - DF, Zip Code 70.040-020, BRAZIL}
\email{pimenta@ifsc.usp.br}

\begin{abstract} The $sl_N$-Onsager algebra has been introduced by Uglov and Ivanov in 1995. In this letter, a FRT presentation of the $sl_N$-Onsager algebra is given, its current algebra and commutative subalgebra are constructed. Certain quotients of the $sl_N$-Onsager algebra are then considered, which produce `classical' analogs of higher rank extensions of the Askey-Wilson algebra. As examples, the cases $N=3$ and  $N=4$ are described in details.
\end{abstract}

\maketitle

\vskip -0.5cm


\vspace{0cm}

\vspace{3mm}

\section{Introduction}
Introduced by A. Zhedanov in \cite{Z91,Z92}, the so-called Askey-Wilson algebra provides an algebraic scheme for the Askey-Wilson polynomials and for the q-Racah polynomials in the finite dimensional case.
Since its introduction, it has been understood that the Askey-Wilson algebra is connected to the double affine Hecke algebra of type $(C_1^{\vee},C_1)$ \cite{K,T12,M,KM}, the theory of Leonard pairs \cite{T87,NT,TV} and $U_q(sl_2)$ \cite{GZ,WZ}. In the context of mathematical physics and the theory of quantum integrable systems,  it has found fruitful applications. For instance, on one hand the Askey-Wilson algebra provides a powerful tool in the Racah problem for $U_q(sl_2)$: given a 3-fold tensor product of irreducible unitary $U_q(sl_2)$ representations $V$, the intermediate Casimir operators are identified with the basic generators of the Askey-Wilson algebra\footnote{More recently, this has been extended to the universal Askey-Wilson algebra introduced in \cite{T11}, see \cite{HH15,HH17}.}  \cite{GZ}. Recall that the inner products (overlap coefficients) between the coupled bases of $V$ are the so-called Racah-Wigner coefficients. The Askey-Wilson algebra allows to determine the explicit expressions for these coefficients in terms of the Askey-Wilson polynomials \cite{GZ}. On the other hand, the generators of the Askey-Wilson algebra can be realized as twisted primitive elements of $U_q(sl_2)$ \cite{Z91}, thus giving one of the simplest realization  of the q-Onsager algebra \cite{Bas2} which arises in the analysis of the XXZ spin chain with non-diagonal boundary conditions \cite{BK2}.
\vspace{1mm}

A well-known presentation of the Askey-Wilson algebra is given in terms 
of three generators satisfying certain quadratic relations \cite{Z91,Z92}. A second presentation of the Askey-Wilson algebra is obtained as a  quotient of the tridiagonal algebra (q-Onsager algebra) \cite{Ter93,Ter99} by certain cubic relations, the so-called Askey-Wilson relations \cite{TV}. Besides these two presentations, a third one is given in terms of solutions  (called Sklyanin's operators) of the reflection equation algebra associated with $U_q(\widehat{sl_2})$ \cite{WZ,Bas2}. \vspace{1mm}

At the specialization $q=1$, the analog of these three presentations are known. The corresponding `classical' limit of the Askey-Wilson algebra admits a first presentation that 
is isomorphic to $sl_2$, and a second presentation as a quotient of the Onsager algebra \cite{Ons44} by the simplest example of 
the so-called Davies' relations \cite{Davies} (see \cite{BC} for details). A third presentation has been recently identified \cite{BBC}, given by solutions of a non-standard Yang-Baxter algebra 
associated with the classical r-matrix of $sl_2$. Importantly, this latter presentation of the `classical' analog of the Askey-Wilson algebra is derived from the Faddeev-Reshetikhin-Taktadjan (FRT) 
presentation of the Onsager algebra, see \cite{BC}.
\vspace{1mm}

Generalizations of the Askey-Wilson algebra is an active field of investigation. Different approaches have been considered in the literature. For instance, in \cite{PW} a `higher rank' Askey-Wilson algebra is introduced and studied. In that paper, it is important to stress that the term `higher rank' refers to the construction of the hidden algebra associated with $N-$fold  tensor products\footnote{For $N=3$, one recovers the usual Askey-Wilson algebra.} of irreducible representations of $U_q(sl_2)$. In \cite[Section 2.4]{BK2}, a different approach is considered: `generalized' Askey-Wilson algebras are introduced, characterizing tensor product evaluation representations of the q-Onsager algebra (that describe certain quotients of a coideal subalgebra of $U_q(\widehat{sl_2})$).  Also, let us mention the non-standard quantum algebra $U'_q(so_n)$ introduced in \cite{GK} (for a review, see \cite{K01}) which can be viewed as a higher rank extension of the Askey-Wilson algebra associated with the {\it finite} Lie algebra $so_n$. However, to our knowledge, the construction of extensions of the Askey-Wilson algebra associated with {\it higher rank affine Lie algebras} has not been considered yet in the literature. This problem is adressed in this letter, in the simplified case of $q=1$. \vspace{1mm}

In view of the richness of the Askey-Wilson algebra, there are several motivations for identifying higher rank extensions associated with higher rank Lie algebras $g$. For instance, these extensions are expected to provide an algebraic framework for multivariate orthogonal polynomials generalizing the Askey-Wilson ones. Also, the construction of a higher rank Askey-Wilson algebra associated with $U_q(g)$ should provide a useful tool to derive Racah-Wigner coefficients associated with 3-fold tensor products of irreducible unitary $U_q(g)$ representations. For $g=sl_N$, this problem is recently studied in \cite{NPZ,MMS,R18}. Also,  higher rank Askey-Wilson algebras will correspond to certain quotients of  the generalized $q-$Onsager algebras\footnote{Generalized $q-$Onsager algebras are known to be isomorphic to certain coideal subalgebras of quantum affine algebras. The isomorphism is given in \cite{BB1}. For an interpretation within the theory of quantum symmetric pairs associated with $U_q(\widehat{g})$, see \cite{Kolb}.} introduced in \cite{BB1}. In this context, they should determine the generating functions relating higher rank extensions of certain products of q-Pochhammer functions that are known for the $U_q(sl_2)$ case, see \cite{Ro}.  
\vspace{1mm}

The main results of this letter are the following: exploiting the framework of \cite{BBC} specialized to the  case of the affine Lie algebra $a_{N-1}^{(1)}$, we derive three different presentations of a higher rank generalization of the Askey-Wilson  algebra, see Section 3.  On the path, we also  identify a FRT presentation of the $sl_N$-Onsager algebra introduced in \cite{UI} (see also \cite{S18}). This FRT presentation is a powerful tool to derive in a systematic way the current algebra and mutually commuting elements of the $sl_N$-Onsager algebra that should find  applications in related integrable systems. For $q\neq 1$, see e.g. \cite{FK}.\vspace{1mm}

The text is organized as follows. In Section 2, the FRT presentation of the affine Lie algebra $a_{N-1}^{(1)}$ denoted $\cT$ is first introduced, see Definition \ref{def:T}. For the r-matrix associated with   $a_{N-1}^{(1)}$ \cite{J}, automorphisms of $\cT$, denoted $\theta_1,\theta_2$, are given  in Proposition \ref{prop:tw}. Then, the FRT presentation for the fixed point subalgebra of $\cT$ with respect to the automorphism $\theta_1$ is given in Proposition \ref{progO}. It provides an alternative presentation of the so-called $sl_N$-Onsager algebra (with original defining relations (\ref{eq:gO2})-(\ref{eq:gO1}) or equivalently Definition \ref{p1}) introduced by Uglov and Ivanov in \cite{UI}, denoted $\cO(a_{N-1}^{(1)})$. The FRT presentation allows to derive the corresponding current algebra, see Proposition \ref{prop1}, and the explicit elements in a commutative subalgebra of the $sl_N$-Onsager algebra, see Proposition \ref{pr:a2} and (\ref{eq:charges}). In Section 3, certain quotients of the $sl_N$-Onsager algebra are considered. These quotients provide natural extensions of the `classical' limit $q=1$ of the Askey-Wilson algebra. The cases of the $sl_3$ and $sl_4$-Askey-Wilson algebra, denoted respectively $AW(a_{2}^{(1)})$ and $AW(a_{3}^{(1)})$, are described in details: three different presentations are identified. For generic values of $N$, the FRT presentation of the $sl_N$-Askey-Wilson algebra $AW(a_{N-1}^{(1)})$ is given.
Concluding remarks follow in the last Section.

\section{The $sl_N$-Onsager algebra revisited}
In this section, firstly we recall the defining relations of the affine Lie algebra $a_{N-1}^{(1)}$ and its presentation within the framework of the classical 
Yang-Baxter algebra $\cT$. Two automorphisms of $\cT$ are introduced. Considering the fixed point subalgebra of $\cT$ with respect to one of these automorphisms, 
denoted $\cB$, a FRT presentation for the  $sl_N$-Onsager algebra is obtained. 
As an application, the corresponding current algebra and  mutually commuting quantities are derived.

\subsection{Classical Yang-Baxter algebra for the affine Lie algebra $a_{N-1}^{(1)}$ \label{sec:aff}}

The affine Lie algebra $a_{N-1}^{(1)}$ is generated by $\{c,e_{ij}^{(n)}\ | \ n\in \ZZ, 1\leq i,j\leq  N  \}$ subject to the defining relations, 
for $n \in\ZZ$, $1\leq i,j  \leq N$ and $1\leq k,l  \leq N$,
\begin{eqnarray}
&&\left[e_{ij}^{(m)},e_{kl}^{(n)}\right]=
\delta_{jk}e_{il}^{(m+n)}-\delta_{il}e_{kj}^{(m+n)}+m\,c\,\delta_{m+n 0}\left(\delta_{il}\delta_{jk}-\frac{1}{N}\delta_{ij}\delta_{kl}\right)\label{hrelation}\ ,\\
&&\sum_{i=1}^N e_{ii}^{(n)}=0\quad\text{and}\qquad [c,e_{ij}^{(n)}]=0\ .\label{hrelationf}
\end{eqnarray}
Note that we may introduce the Cartan generators $h_i^{(n)}$ defined by
\begin{equation}
 h_i^{(n)}=e^{(n)}_{ii}-e^{(n)}_{i+1 \ i+1}\;.\label{eq:cart}
\end{equation}
Thanks to the first relation in \eqref{hrelationf}, we can invert the previous relation and express $e^{(n)}_{ii}$ in terms of $h_i^{(n)}$
\begin{equation}
 e^{(n)}_{ii}=-\sum_{j=1}^{i-1} \frac{j}{N} h_j^{(n)} + \sum_{j=i}^{N-1} \frac{N-j}{N} h_j^{(n)}\;.
\end{equation}

For any affine Lie algebra, an alternative presentation is given by a classical Yang-Baxter algebra  \cite{Sk2}. The basic ingredient in this presentation is the so-called classical r-matrix which is defined as follows.
\begin{defn}  \label{def:cybe}
The matrix
  $r(x)\in End(\CC^N\otimes \CC^N)$ is called a classical r-matrix if it satisfies the skew-symmetric condition $r_{12}(x)=-r_{21}(1/x)$ and the classical Yang-Baxter equation
 \begin{equation}\label{eq:CYBE}
  [\ r_{13}(x_1/x_3)\ , \ r_{23}(x_2/x_3)\ ]=[\ r_{13}(x_1/x_3)+  r_{23}(x_2/x_3)\ , \ r_{12}(x_1/x_2)\ ]\;
 \end{equation}
for any $x_1,x_2,x_3$.  Here, $r_{12}(x) = r(x)\otimes \mathbb{I}$ , $r_{23}(x) =\mathbb{I}\otimes  r(x) $  and so on.
 \end{defn}
 \vspace{1mm}

For all non-exceptional affine Lie algebras, the classical r-matrices are obtained in \cite{J}. In particular,
for the affine Lie algebra $a_{N-1}^{(1)}$, the r-matrix is given by:
\begin{equation}\label{def:rJimbo}
 r(x)= \frac{1+x}{1-x} \left(\sum_{\genfrac{}{}{0pt}{1}{i=1}{}}^N    E_{ii}\otimes E_{ii}-\frac{1}{N}\mathbb{I}\otimes\mathbb{I}\right)
+\frac{2}{1-x} \sum_{\genfrac{}{}{0pt}{1}{i,j=1}{i< j}}^N     E_{ij}\otimes E_{ji}
+\frac{2x}{1-x}\sum_{\genfrac{}{}{0pt}{1}{i,j=1}{i> j}}^N     E_{ij}\otimes E_{ji}\ ,
\end{equation}
where  $\{ E_{ij}\ | \ 1\leq i,j\leq N \}$  denotes the standard basis of $End(\CC^N)$ given by $N\times N$ matrices with components $(E_{ij})_{kl}=\delta_{ik}\delta_{jl}$.\vspace{1mm}

Given the r-matrix $r(x)$, the so-called classical Yang-Baxter algebra 
$\cT$ with central extension $c$ is now introduced. Note that this algebra is a Lie algebra which can be understood as a classical analog of the algebra introduced in \cite{RS}.
\vspace{1mm}

\begin{defn}\label{def:T} Let  $r(x)\in End(\CC^N\otimes \CC^N)$ be the classical r-matrix defined by (\ref{def:rJimbo}).  $\cT$ is defined as the Lie algebra with generators
$\{c, e_{ij}^{\pm [\ell]} \ |\ 1\leq i,j\leq N,\ell\in\ZZ \}$.
Introduce the following elements of $End(\CC^N)\otimes a_{N-1}^{(1)}$:
\ben
T^{+}(x)&=&\sum_{i=1}^N E_{ii}\otimes e_{ii}^{(0)}+2 \sum_{1\leq i<j\leq N} E_{ij}\otimes e_{ji}^{(0)}
+2\sum_{n\geq 1}x^n \sum_{1\leq i,j\leq N} E_{ij}\otimes e_{ji}^{(n)}\ ,\label{eq:Tp}\\
T^{-}(x)&=&-\sum_{i=1}^N E_{ii}\otimes e_{ii}^{(0)}-2 \sum_{1\leq j<i\leq N} E_{ij}\otimes e_{ji}^{(0)}
-2\sum_{n\geq 1}x^{-n} \sum_{1\leq i,j\leq N} E_{ij}\otimes e_{ji}^{(-n)}\ .\label{eq:Tm}
\een
The defining relations \eqref{hrelation}-\eqref{hrelationf}  are equivalent to:
\begin{eqnarray}
\null[ T^{\pm}(x),c]&=&0\;,\label{Tg}\\
\null[ T_1^{\pm}(x),T^{\pm}_2(y)]&=&[ T_1^{\pm}(x)+T^{\pm}_2(y),r_{12}(x/y)]\;,\label{rpp}\\
\null[ T_1^{+}(x),T^{-}_2(y)]
&=&[ T_1^{+}(x)+T^{-}_2(y),r_{12}(x/y)]-2c\,r'_{12}(x/y)x/y\;,\label{rpm}\\
tr_1 T_1^\pm(x)=0\label{tr0}\;,
\end{eqnarray}
where $r'(x)$ denotes the derivative of $r(x)$ with respect to $x$.
\end{defn}

\begin{example} Explicitly, relation \eqref{eq:Tp} becomes, for $N=3$,
\ben
T^{+}(x)&=&
\left(
\begin{array}{ccc}
 e_{11}^{(0)} & 2e_{21}^{(0)} & 2e_{31}^{(0)} \\
 0 &   e_{22}^{(0)} & 2e_{32}^{(0)} \\
 0 & 0 &   e_{33}^{(0)} \\
\end{array}
\right)+2\sum_{n\geq	 1}x^n\left(
\begin{array}{ccc}
 e_{11}^{(n)} & e_{21}^{(n)} & e_{31}^{(n)} \\
 e_{12}^{(n)} &   e_{22}^{(n)} & e_{32}^{(n)} \\
 e_{13}^{(n)} & e_{23}^{(n)} &  e_{33}^{(n)} \\
\end{array}
\right)\,.
\een
\end{example}
\vspace{2mm}

Automorphisms of the Lie algebra $\cT$ are now considered. They play a central role in the construction of fixed point subalgebras of $\cT$ that will be discussed in the next subsection. Note that the proof of the following proposition essentially follows \cite[Proposition 2.1]{BBC}. 
\begin{prop}\label{prop:tw}
The map $\theta_1:End(\CC^N)\otimes a_{N-1}^{(1)} \rightarrow End(\CC^N)\otimes a_{N-1}^{(1)}$ such that\footnote{Here $(.)^t$ stands for the transposition in $End(\CC^N)$.} 
\begin{eqnarray}\label{eq:auU}
 T^\pm(x)\mapsto U\ T^\mp((-1)^N/x)^t\ U \quad,\qquad c \mapsto -c\  \quad\text{with} \qquad U=\sum_{j=1}^N (-1)^j E_{jj}
\end{eqnarray}
provides an involutive automorphism of $a_{N-1}^{(1)}$.
 
For $N$ even, the map $\theta_2:End(\CC^N)\otimes a_{N-1}^{(1)}\rightarrow End(\CC^N)\otimes a_{N-1}^{(1)}$ such that
\begin{eqnarray}\label{eq:auV}
 T^\pm(x)\mapsto  V(x)\ T^\mp(1/x)^t\ V(x)^{-1} \ \mp \ c x V'(x)V(x)^{-1}, \quad c \mapsto -c\   \ 
\end{eqnarray}
with
\beqa
V(x)=\sum_{j=1}^\frac{N}{2}\left(\frac{1}{\sqrt x}E_{j,\frac{N}{2}+j}+\epsilon  {\sqrt x} E_{\frac{N}{2}+j,j}\right) \quad \mbox{for $\epsilon=+1$ or $\epsilon=-1$}\  \label{eq:Undiag}
\eeqa 
 provides an involutive automorphism of $a_{N-1}^{(1)}$. 

\end{prop}
\begin{proof}
We show that the images of $T^{\pm}(x)$ by the automorphism $\theta_1$ satisfy relations \eqref{Tg}-\eqref{tr0} by using that
\begin{eqnarray} 
 U_1 U_2 r_{12}(x/y)=r_{12}(x/y)U_1 U_2\;.\label{eqU}
\end{eqnarray}
The involutive property is proven by remarking that $U^t= U$. The proof is similar for $\theta_2$.
\end{proof} 

\vspace{1mm}

The action of each automorphism on the generators of $a_{N-1}^{(1)}$ is now described. 
From (\ref{eq:auU}), the automorphism $\theta_1$ is such that 
\beqa
\theta_1( e_{ij}^{(n)})= (-1)^{Nn+i+j+1}e_{ji}^{(-n)} \quad \mbox{and}\quad \theta_1(c)=-c.\label{eq:aut1}
\eeqa
From (\ref{eq:auV}), the automorphism $\theta_2$ is such that (we recall that $N$ must be even in this case)
\beqa
 && \quad \theta_2(e_{ij}^{(n)})=\begin{cases}
                          -e_{\bar \jmath\bar \imath}^{(-n)} +\bar\alpha_i\frac{c}{2}\delta_{ij}\delta_{n0}&\text{if }  1\leq i,j\leq N/2 \quad\text{or}\quad N/2+1\leq i,j\leq N  \\
                          -\epsilon\ e_{\bar \jmath\bar \imath}^{(-n+1)}&\text{if }  1\leq i\leq N/2,\ N/2+1\leq j\leq N \\
                          -\epsilon\ e_{\bar \jmath\bar \imath}^{(-n-1)}&\text{if }  N/2+1\leq i\leq N,\ 1\leq j\leq N/2\ , 
                         \end{cases}
 \ \mbox{and}\ \theta_2(c)= -c\,,\label{eq:aut2}
\eeqa
where $\bar \imath=i+N/2$ if $1\leq i\leq N/2$, $\bar \imath=i-N/2$ if $N/2+1\leq i\leq N$, $\bar \alpha_i=1$ if $1\leq i\leq N/2$ and $\bar \alpha_i=-1$ if $N/2+1\leq i\leq N$.

\subsection{The $sl_N$-Onsager algebra}\label{sec:Ons}

 The $sl_{N}$-Onsager algebra has been introduced  in \cite{UI} as  one possible higher rank generalization of the Onsager algebra.  It is known that the Onsager algebra \cite{Ons44}  is isomorphic to the fixed point subalgebra of the affine Lie algebra $\widehat{sl_2}$ \cite{Perk1,Davies,Roan}. By analogy,  
the $sl_{N}$-Onsager algebra is defined as the fixed point  subalgebra of $a_{N-1}^{(1)}$  under the action of the involutive automorphism $\theta_1$ \eqref{eq:aut1}, see \cite{UI,S18} for details. Therefore, in the text below the $sl_{N}$-Onsager algebra is denoted $\cO(a_{N-1}^{(1)})$. Let us introduce the following generators of the subalgebra of $a_{N-1}^{(1)}$
\begin{eqnarray}
 &&B_{ij}^{(n)}=e_{ij}^{(n)}+(-1)^{i+j+1+nN} e_{ji}^{(-n)}\; \quad \mbox{for}\quad 1\leq i, j\leq N \label{eq:Bij}
\end{eqnarray}
which are left invariant under the automorphism $\theta_1$ \eqref{eq:aut1}.
\begin{defn}\label{p1}
The algebra $\cO(a_{N-1}^{(1)})$ is generated by $\{ B_{ij}^{(n)} |n \in\ZZ\}$  subject to the following relations:
\begin{eqnarray}
 &&[B_{ij}^{(m)},B_{kl}^{(n)}]  =  \delta_{jk}B_{il}^{(m+n)} -
\delta_{il}B_{kj}^{(m+n)}+\delta_{ik}(-1)^{i+j+1+mN}B_{jl}^{(n-m)}
-\delta_{jl}(-1)^{i+j+1+mN}B_{ki}^{(n-m)}\ , \label{Bdef1}\\
&&B_{ij}^{(n)}=(-1)^{i+j+1+nN}B_{ji}^{(-n)}\;, \label{Bdef2}\\
&&\sum_{i=1}^N B_{ii}^{(n)}=0\;.\label{Bdef3}
\end{eqnarray}
\end{defn}

The original presentation of the $sl_N$-Onsager algebra given in \cite{UI} is easily identified. Introduce: 
\begin{eqnarray}
 &&A_{ij}^{(n)}=B_{ij}^{(n)}\quad\text{for }1\leq i\neq j\leq N\ ,\\
 &&G_i^{(n)}=B_{ii}^{(n)}-B_{i+1\ i+1}^{(n)}=h_{i}^{(n)}+(-1)^{nN+1}h_{i}^{(-n)}\qquad\text{for }\ \ 1\leq i\leq N-1\,,
\end{eqnarray}
where $h_{i}^{(n)}$ is defined by \eqref{eq:cart}. In terms of the generators $\{A_{ij}^{(n)},G_{ij}^{(n)}\}$, the defining relations (\ref{Bdef1})-(\ref{Bdef3}) now read:
 \ben
[A_{ij}^{(m)},A_{kl}^{(n)}] & = & \delta_{jk}A_{il}^{(m+n)} -
\delta_{il}A_{kj}^{(m+n)}  \label{eq:gO2} \\
& + & \delta_{ik}(-1)^{i+j+1+mN}\theta (j<l)A_{jl}^{(n-m)} +
\delta_{ik}(-1)^{i+l+nN}\theta (l<j)A_{lj}^{(m-n)}  \nonumber \\
& + & \delta_{jl}(-1)^{k+l+1+nN}\theta (i<k)A_{ik}^{(m-n)} 
+\delta_{jl}(-1)^{i+l+mN}\theta (k<i)A_{ki}^{(n-m)}  \nonumber \\
& + & \delta_{ik} \delta_{jl} (-1)^{i+j+1+nN}\sum_{s=i}^{j-1} G_s^{(m-n)}, \; \; \;
\; m\geq n \  ,   \nonumber \\ \mbox{}
[G_i^{(m)},A_{kl}^{(n)}] & = & ( \delta_{ik}- \delta_{ki+1} - \delta_{li} +
\delta_{li+1})(A_{kl}^{(m+n)}-(-1)^{mN}A_{kl}^{(n-m)}) \ ,\label{eq:gO1}\\ \mbox{}
 [G_i^{(m)}, G_j^{(n)}] & = & 0\ .\nonumber
\een
Here $\theta (x)$ is such that $\theta (x) \; = \; 1\ (0)$ if $x$ is true\ (false).\vspace{1mm}

It is important to mention that an alternative presentation of  the $sl_N$-Onsager algebra (equivalently $\cO(a_{N-1}^{(1)})$)  for $N>2$ is known in terms of $N$ generators $\be_i$, $i=1,...,N$ subject to the relations \cite{UI}:
 \begin{eqnarray}
\big[\be_i,\big[\be_i,\be_j\big]\big]&=&\be_j \quad \mbox{if} \quad i,j \quad \mbox{are adjacent},\label{eq:OAn}\\
\big[\be_i,\be_j\big]&=&0 \qquad \mbox{otherwise}.
 \end{eqnarray}
In particular, one has the identification:
\beqa
 \be_{i}=A_{ii+1}^{(0)}\quad  \mbox{for}\quad  i=1,...,N-1, \qquad \be_{N}=A_{1N}^{(-1)}.
\eeqa
For $N=2$, the corresponding presentation of the original Onsager algebra (i.e. $\cO(a_{1}^{(1)})$) is given in terms of two generators $\{\be_{0}, \be_{1}\}$ satisfying the so-called Dolan-Grady relations \cite{DG82}.\vspace{1mm}

\subsection{FRT presentation of  $\cO(a_{N-1}^{(1)})$}
Following \cite{BBC}, the algebra $\cO(a_{N-1}^{(1)})$ admits an alternative presentation given by a non-standard classical  Yang-Baxter algebra as we now show. The basic ingredient to construct this presentation is  a non-skew symmetric r-matrix which is defined as follows.
\begin{defn}\label{def:nscybe} The matrix $\overline{r}(x,y)\in End(\CC^N\otimes \CC^N)$ is called a non-standard classical r-matrix  if it is a solution of 
the non-standard classical Yang-Baxter equation
 \begin{equation}\label{eq:nsCYBE}
  [\ \overline{r}_{13}(x_1,x_3)\ , \ \overline{r}_{23}(x_2,x_3)\ ]=[\ \overline{r}_{21}(x_2,x_1)\ , \ \overline{r}_{13}(x_1,x_3)\ ]+[\ \overline{r}_{23}(x_2,x_3)\ , \ \overline{r}_{12}(x_1,x_2)\ ]\;
 \end{equation}
for any $x_1,x_2,x_3$. 
\end{defn}
Note that the non-standard classical Yang-Baxter equation can be understood as a generalization of the classical Yang-Baxter equation (\ref{eq:CYBE}). If $\overline{r}_{12}(x,y)=r_{12}(x/y)$ and   $\overline{r}_{12}(x,y)=-\overline{r}_{21}(y,x)$, the non-standard classical Yang-Baxter equation reduces to (\ref{eq:CYBE}). \vspace{1mm}

For all non-exceptional affine Lie algebras, given the r-matrix $r(x)$ and the knowledge of automorphisms of $\cT$, solutions of the non-standard classical Yang-Baxter equation are easily constructed. For the affine Lie algebra $a_{N-1}^{(1)}$ and the automorphism $\theta_1$ of Proposition \ref{prop:tw}, 
\beqa
\overline{r}_{12}(x,y)=r_{12}(x/y)+U_1 r_{12}^{t_1}((-1)^N/(xy))U_1^{-1}
\eeqa
satisfies (\ref{eq:nsCYBE}). Explicitly, using (\ref{def:rJimbo}), (\ref{eq:auU}) one gets:
\beqa\label{eq:rbarO}
\overline{r}(x,y)&=& -\left(\frac{x+y}{x-y} + \frac{xy+(-1)^N}{(-1)^N-xy}\right) \left(\sum_{\genfrac{}{}{0pt}{1}{i=1}{}}^N    E_{ii}\otimes E_{ii}-\frac{1}{N}\mathbb{I}\otimes\mathbb{I} \right) 
-\frac{2y}{x-y} \sum_{\genfrac{}{}{0pt}{1}{i,j=1}{i< j}}^N     E_{ij}\otimes E_{ji} \nonumber\\
&&-\frac{2x}{x-y}\sum_{\genfrac{}{}{0pt}{1}{i,j=1}{i> j}}^N     E_{ij}\otimes E_{ji}+\frac{2xy}{xy-(-1)^N} \sum_{\genfrac{}{}{0pt}{1}{i,j=1}{i< j}}^N   (-1)^{i+j}  E_{ji}\otimes E_{ji} +\frac{2(-1)^N}{xy-(-1)^N}\sum_{\genfrac{}{}{0pt}{1}{i,j=1}{i> j}}^N  (-1)^{i+j}  E_{ji}\otimes E_{ji}. \label{eq:rbar}
\eeqa

Given the non-standard r-matrix $\overline{r}(x,y)$ (\ref{eq:rbar}), the non-standard classical Yang-Baxter algebra 
$\cB$  associated with $a_{N-1}^{(1)}$  is now introduced. 
\begin{defn}[see also \cite{BBC}]\label{def:B}  $\cB$ is the Lie algebra with generators 
$\{B_{ji}^{(n)}\ |\  \ 1\leq i,j\leq N, n\in\ZZ_{\geq 0} \}$. Introduce:
\ben
B(x)&=& 2\sum_{1\leq i<j\leq N} E_{ij}\otimes B_{ji}^{(0)} +2\sum_{1\leq i,j\leq N}\sum_{n\geq 1}x^n E_{ij}\otimes B_{ji}^{(n)}\ .\label{eq:Bx}
\een
The defining relations are:
\beqa
&& [\ B_{1}(x)\ , \ B_{2}(y)\ ]=[\ \overline{r}_{21}(y,x)\ , \ B_{1}(x)\ ]+[\ B_{2}(y)\ , \ \overline{r}_{12}(x,y)\ ]\; ,\label{eq:Al}\\
&& tr_1B_{1}(x)=0\,,\label{eq:Al2}
\eeqa
with (\ref{eq:rbar}).
\end{defn}

Using the results of Section \ref{sec:aff}, we are now ready to provide the FRT presentation of the $sl_{N}$-Onsager algebra $\cO(a_{N-1}^{(1)})$. Note that the construction presented below can be seen as a classical limit of the presentation of the (quantum) twisted Yangians given in \cite{MNO,MRS}.
\begin{prop}\label{progO} Define the element of  $End(\CC^N)\otimes \cO(a_{N-1}^{(1)})$:
 \begin{equation}\label{eq:Bxg}
 B(x)= T^+(x) + \theta_1(T^+(x))=T^+(x) + U\ T^-((-1)^N/x)^t\ U^{-1} \;.
\end{equation}
The relations \eqref{eq:gO2}-\eqref{eq:gO1} are equivalent to the relations \eqref{Bdef1}-\eqref{Bdef3}.
\end{prop}
\begin{proof}
 Using the explicit expressions \eqref{eq:Tp}-\eqref{eq:Tm} of $T^\pm(x)$, from the definition (\ref{eq:Bxg})  one has  (\ref{eq:Bx}) with  (\ref{eq:Bij}).  This implies (\ref{Bdef2}) as well as (\ref{Bdef3}) using the first equation in (\ref{hrelationf}).  Then, inserting  $B(x)$ given by \eqref{eq:Bxg} with all the independent generators of $\cO(a_{N-1}^{(1)})$ into the relation \eqref{eq:Al} with (\ref{eq:rbar}), one finds all the commutation relations (\ref{Bdef1}).
\end{proof}

For $N=2$, a current algebra presentation for the Onsager algebra has been derived in \cite{BBC}. For $N>2$, we can use the results of previous sections to derive the current presentation of the $sl_N$-Onsager algebra.
\begin{prop}\label{prop1} The algebra $\cO(a_{N-1}^{(1)})$ admits the following current presentation. Define the currents:
\begin{eqnarray}
B_{ij}(x)=\begin{cases}
            \displaystyle  2\sum_{n\geq 0}  x^nB_{ij}^{(n)}&\text{for}\quad i>j\\
             \displaystyle   2\sum_{n\geq 1}  x^nB_{ij}^{(n)}&\text{otherwise} 
             \end{cases} \ .
\end{eqnarray}
The defining relations for the  currents are given by:
\beqa
\qquad [B_{ij}(x),B_{kl}(y)]&=&\frac{2\delta_{jk}}{x-y}\Big( B_{il}(x)\big(xH(k-l)+yH(l-k)\big)-B_{il}(y)\big(yH(i-j)+xH(j-i)\big)   \Big)\\
&-&\frac{2\delta_{il}}{x-y}\Big( B_{kj}(x)\big(xH(k-l)+yH(l-k)\big)-B_{kj}(y)\big(yH(i-j)+xH(j-i)\big)   \Big)\nonumber\\
&-&\frac{2\delta_{ik}}{xy-(-1)^N}\Big( B_{lj}(x)(-1)^{k+l}\big(xyH(l-k)+(-1)^NH(k-l)\big)\nonumber\\
&&\hspace{3cm}-B_{jl}(y)(-1)^{i+j}\big(xyH(j-i)+(-1)^NH(i-j)\big)   \Big)\nonumber\\
&+&\frac{2\delta_{jl}}{xy-(-1)^N}\Big( B_{ik}(x)(-1)^{k+l}\big(xyH(l-k)+(-1)^NH(k-l)\big)\nonumber\\
&&\hspace{3cm}-B_{ki}(y)(-1)^{i+j}\big(xyH(j-i)+(-1)^NH(i-j)\big)   \Big)\nonumber
\eeqa
where $H$ is the step function: $H(i)=1$ if $i>0$, $H(i)=1/2$ if $i=0$ and $H(i)=0$ if $i<0$. 
\end{prop}
\begin{proof} By straightforward calculations one derives the commutation relations between the currents from the non-standard classical Yang-Baxter algebra (\ref{eq:Al}) with $ B(x)=\sum_{1\leq i,j\leq N} E_{ij}\otimes B_{ji}(x)$.
%
%\begin{equation}
 %B(x)=\sum_{1\leq i,j\leq N} E_{ij}\otimes B_{ji}(x)\;.
%\end{equation}
\end{proof}

\subsection{Commutative subalgebra of  $\cO(a_{N-1}^{(1)})$ \label{sec:CS}}
In the context of integrable systems, the explicit construction of mutually commuting conserved quantities is essential. For models generated from the Onsager algebra, elements in the commutative subalgebra of the Onsager algebra have been constructed in  \cite{Ons44,DG82,Davies}. Recently, it was shown that a generating function for all mutually conserved quantities is easily derived using the FRT presentation of the Onsager algebra \cite{BBC}. By analogy, a generating function for elements in a commutative subalgebra of $\cO(a_{N-1}^{(1)})$ is now given.
\begin{prop}\label{pr:a2} Define
\ben
M(x)&=&\left(x-\frac{(-1)^N}{x}\right) \sum_{i=1}^{N}\mu_i E_{ii} + \sum_{\genfrac{}{}{0pt}{1}{i,j=1}{i< j}}^N\left[ \left(\kappa_{ij}+\frac{\kappa_{ij}^*}{x}\right)  E_{ij}
-  (-1)^{i+j}\left(\kappa_{ij}+(-1)^N x\kappa_{ij}^*\right)  E_{ji}\right]\label{eq:exM}
\een
where $\mu_i, \kappa_{ij}$ and $\kappa_{ij}^*$  are free scalar parameters. Then,
\begin{equation}\label{tb2}
 b(x)=tr M(x) B(x)
\end{equation}
satisfies
 \begin{equation}
  [ b(x)\ ,\ b(y) ]=0\;.\label{eq:bbcom}
 \end{equation}
\end{prop}

\begin{proof}
In \cite{BBC}, it is shown that if
 \begin{equation}\label{eq:reD}
  [tr_1 ( \overline{r}_{12}(x,y) M_1(x) ) \ ,\ M_2(y) ]=0\; 
 \end{equation}
then \eqref{eq:bbcom} holds.
To show that (\ref{eq:exM})  satisfies (\ref{eq:reD}), we first compute $tr_1 ( \overline{r}_{12}(x,y) M_1(x) )$. 
By straightforward calculations, one gets:
\beqa
tr_1 ( \overline{r}_{12}(x,y) M_1(x) )&=& -\left(\frac{x+y}{x-y} + \frac{xy+(-1)^N}{(-1)^N-xy}\right) \left( x-\frac{(-1)^N}{x}\right) \left(\sum_{\genfrac{}{}{0pt}{1}{i=1}{}}^N   \mu_i E_{ii}-\frac{1}{N} \left(\sum_{i=1}^N\mu_i \right)\mathbb{I}\right) \nonumber\\
&& + \sum_{i< j}^N  \underbrace{\left( \left(\frac{2(-1)^N}{(-1)^N-xy}\right) (\kappa_{ij} + (-1)^N x \kappa_{ij}^*)  
+ \left( \frac{2x}{y-x}\right)(\kappa_{ij} + \kappa_{ij}^*/x)  \right) }_{W_{ij}}E_{ij} \nonumber\\
&&+\sum_{i< j}^N \underbrace{\left(\left( \frac{2xy}{xy-(-1)^N} \right)(-1)^{i+j} (\kappa_{ij} +{{(-1)^N}} \kappa_{ij}^*/x)  + \left(\frac{2y}{x-y}\right)(-1)^{i+j}(\kappa_{ij} + (-1)^N x \kappa_{ij}^*)\right)}_{V_{ij}} E_{ji} .\nonumber
\eeqa

Define
\beqa
U_i=-\left(\frac{x+y}{x-y} + \frac{xy+(-1)^N}{(-1)^N-xy}\right) \left( x-\frac{(-1)^N}{x}\right)\mu_i\ \nonumber
\eeqa
and introduce the notation:
\beqa
M(y)= \sum_{i=1}^N  \underbrace{\left( y-\frac{(-1)^N}{y}\right)}_{A_i}\mu_i \ E_{ii} + \sum_{i<j}^N \left( \underbrace{(\kappa_{ij} +  \kappa_{ij}^*/y)}_{B_{ij}} E_{ij} + \underbrace{(-1)^{i+j+1}(\kappa_{ij} + (-1)^N y \kappa_{ij}^*)}_{C_{ij}}  E_{ji} \right)\ .\nonumber
\eeqa
In terms of the functions $A_i,B_{ij},C_{ij},U_i,V_{ij},W_{ij}$, the l.h.s. of (\ref{eq:reD}) reads:
\beqa
[tr_1 ( \overline{r}_{12}(x,y) M_1(x) ) \ ,\ M_2(y) ]&=& \sum_{i< j}^N \underbrace{\left((U_i-U_j)B_{ij} + W_{ij}(A_j-A_i)\right)}_{=0} E_{ij} \nonumber\\
&&+ \sum_{i< j}^N \underbrace{\left((U_j-U_i)C_{ij} + V_{ij}(A_i-A_j)\right)}_{=0} E_{ji} \nonumber\\
&&+ \sum_{i< j<k}^N \underbrace{\left(V_{ji}B_{jk} - V_{kj}B_{ij} + W_{ij}C_{kj} - W_{jk}C_{ji}\right)}_{=0} E_{ik} \nonumber\\
&&+\sum_{i< j<k}^N \underbrace{\left(V_{ji}C_{kj} - V_{kj}C_{ji} + W_{ij}B_{jk} - W_{jk}B_{ij}\right)}_{=0} E_{ik}\ \nonumber\\
&=&0.\nonumber
\eeqa
\end{proof}

From the previous results, we obtain a generalization of the result of Uglov and Ivanov \cite[eq. (10)]{UI}. Namely, expanding the generating function given by $b(x)$ (\ref{tb2}) in $x$ for (\ref{eq:exM}), mutually commuting quantities of the $sl_N-$Onsager algebra are derived. For $n=0,1,2,\dots$, the coefficients of the power series are proportional to:
\beqa
\qquad\quad  I_{n}&=&
\sum _{i=1}^N \left(\kappa _{ij} \left(B_{ij}^{(n)}+(-1)^{i+j+1}
   B_{ji}^{(n)}\right)+\kappa _{i,j}^*
   \left((-1)^{i+j+N+1} B_{ji}^{(n-1)}+B_{ij}^{(n+1)}\right)\right)\label{eq:charges}\\&&\quad+
   \sum _{i=1}^N \mu_i \left((-1)^{N+1}
   B_{ii}^{(n+1)}+B_{ii}^{(n-1)}\right)
\quad \mbox{for}\quad n>1\ ,\nonumber\\
I_{0}&=&
\sum _{i<j}^N \left( \kappa _{ij}(-1)^{i+j+1} 
   B_{ji}^{(0)}+\kappa _{ij}^* B_{ij}^{(1)}\right)+(-1)^{N+1}\sum_{i=1}^N  \mu_i B_{ii}^{(1)}\ ,
\nonumber\\
I_{1}&=& \sum _{i<j}^N \left(\kappa _{ij} \left(B_{ij}^{(1)}+(-1)^{i+j+1}
   B_{ji}^{(1)}\right)+\kappa _{ij}^*
   \left((-1)^{i+j+N+1} B_{ji}^{(0)}+B_{ij}^{(2)}\right)\right)-(-1)^N \sum _{i=1}^N \mu_i B_{ii}^{(2)}.
\nonumber
\eeqa

\section{Higher rank `classical' Askey-Wilson algebras}
In \cite{BC}, certain quotients of the Onsager algebra considered by Davies \cite{Davies} were studied using the framework of the FRT presentation of the Onsager algebra proposed in \cite{BBC}. In particular, to each of these quotients  a generalization of the `classical' analog $(q=1)$ of the Askey-Wilson algebra \cite{Z91,Z92} is associated. By analogy, it is clearly expected that certain quotients of the $sl_N$-Onsager algebra should lead to  classical analogs of higher rank extensions of the Askey-Wilson algebra. Below, we exploit the FRT presentation of the  $sl_N$-Onsager algebra discussed  in the previous section to derive explicit examples for $N=3$ and $N=4$. Higher rank extensions of the Askey-Wilson algebra for arbitrary $N$ are briefly discussed.

\subsection{The $sl_3$-Askey-Wilson algebra $AW(a_{2}^{(1)})$} The starting point is the explicit construction of solutions of the non-standard Yang-Baxter algebra (\ref{eq:Al}), (\ref{eq:Al2}) that correspond to certain quotients of the $sl_3$-Onsager algebra. Besides the defining relations (\ref{Bdef1})-(\ref{Bdef3})  for $N=3$, we are looking for solutions of  (\ref{eq:Al}), (\ref{eq:Al2})  of the form (\ref{eq:Bx}) where certain additional linear relations among the generators $\{B_{ji}^{(n)}\}$ are assumed.
 The simplest non-trivial example that we have found for $N=3$ is displayed in the definition below.
\begin{defn}\label{def:AW2}
 The $sl_3$-Askey-Wilson algebra, denoted $AW(a_{2}^{(1)})$, is generated by $\{\be_i,\bof_i,\bg_j\ | \ i=1,2,3,j=1,2\}$ subject to the defining relations \eqref{eq:Al}, \eqref{eq:Al2} with
 \begin{equation}\label{eq:Baw2}
 B(x)=\frac{2}{\alpha+x-1/x}\begin{pmatrix}
                             \frac{2}{3}\bg_1+\frac{1}{3}\bg_2&\bof_1-x^{-1}\be_1 &\be_3+x^{-1}\bof_3\\
                             -x\be_1-\bof_1 &-\frac{1}{3}\bg_1+\frac{1}{3}\bg_2&\bof_2-x^{-1}\be_2\\
                             \be_3-x\bof_3 & -x\be_2-\bof_2&-\frac{1}{3}\bg_1-\frac{2}{3}\bg_2\ 
                            \end{pmatrix}\ .
\end{equation}
\end{defn}

Inserting (\ref{eq:Baw2}) into the relations \eqref{eq:Al}, one gets equivalently the following defining relations of the $sl_3$-Askey-Wilson algebra:
\begin{eqnarray}
&& [\be_i,\be_j]=\epsilon_{ijk}\bof_k\ ,\label{eq:AW21}\\
&& [\be_i,\bof_j]=\delta_{ij}\bg_i-\epsilon_{ijk}\be_k\ ,\label{eq:AW22}\\
&& [\be_i,\bg_j]=\alpha\be_i-2\bof_i + 3\delta_{ij}(2\bof_i-\alpha\be_i)\ ,\label{eq:AW23}\\
&& [\bof_i,\bof_j]=\epsilon_{ijk}(\bof_k-\alpha \be_k)\ ,\label{eq:AW24}\\
&& [\bof_i,\bg_j]=-\alpha\bof_i-2\be_i + 3\delta_{ij}(\alpha\bof_i+2\be_i)\ ,\label{eq:AW25}\\
&& [\bg_1,\bg_2]=0\ ,\label{eq:AW26}
\end{eqnarray}
where $\epsilon_{ijk}$ is the Levi-Civita symbol (Einstein summation is assumed on the index $k$)  and $\bg_3=-\bg_1-\bg_2$.\vspace{1mm}

By analogy with the $sl_3$-Onsager algebra which admits two presentations  (see Definition \ref{p1} and (\ref{eq:OAn})), an alternative presentation for the  $sl_3$-Askey-Wilson algebra can be identified.
\begin{prop}\label{pro1}
 The $sl_3$-Askey-Wilson algebra $AW(a_{2}^{(1)})$ is isomorphic to the algebra generated by $\{\bar \be_1,\bar \be_2,\bar \be_3\}$ and subject to the following relations:
\begin{eqnarray}
&&\big[\bar \be_i,\big[\bar \be_i,\bar \be_j\big]\big]=\bar \be_j \quad \mbox{for} \quad 1\leq i \neq j\leq 3\ , \label{eq:OA3}\\
&&\Big[ \big[\bar \be_i,\bar \be_j\big],[\bar \be_j,\bar \be_k\big] \Big] + \big[\bar \be_i,\bar \be_k]+\alpha\epsilon_{ijk} \bar \be_j=0  \quad \mbox{for} \quad 1\leq i \neq j\neq k \leq 3 \ .\label{eq:OA3q}
 \end{eqnarray}
\end{prop}
\begin{proof} Let us denote by $\overline{AW}(a_{2}^{(1)})$ the algebra with the defining relations (\ref{eq:OA3}), (\ref{eq:OA3q}). Firstly, let us show that the map $\phi:\overline{AW}(a_{2}^{(1)})\rightarrow AW(a_{2}^{(1)})$ defined by $\phi(\bar \be_i)=\be_i$ is an algebra homomorphism.  
By using the relations of ${AW}(a_{2}^{(1)})$, we show that, for $1\leq i\neq j \leq 3$,  
\begin{equation}
 [\be_i,[\be_i,\be_j]]=\epsilon_{ijk}[\be_i,\bof_k]=\epsilon_{ijk}(\delta_{ik}\bg_i-\epsilon_{ik\ell}\be_\ell)=(\delta_{j\ell}-\delta_{i\ell}\delta_{ij})\be_\ell=\be_j
\end{equation}
and, for $1\leq i \neq j\neq k \leq 3$,
\begin{equation}
 [[\be_i,\be_j],[\be_j,\be_k]]=\epsilon_{ijs}\epsilon_{jk\ell}\epsilon_{s\ell p}(\bof_p-\alpha\be_p)=\epsilon_{ijk}(\bof_j-\alpha\be_j)=-[\be_i,\be_k]-\alpha\epsilon_{ijk}\be_j\;.
\end{equation}
Therefore, $\phi$ is an algebra homomorphism. \vspace{1mm}

Secondly, let us remark that ${AW}(a_{2}^{(1)})$ is generated only by $\{\be_1, \be_2, \be_3\}$. Indeed, using relations \eqref{eq:AW21} and \eqref{eq:AW22}, we can express the generators 
$\bof_1, \bof_2, \bof_3, \bg_1$ and $\bg_2$ as follows
\begin{eqnarray}
 \bof_1=[\be_2,\be_3]\ ,\ \bof_2=[\be_3,\be_1]\ ,\ \bof_3=[\be_1,\be_2]\ ,\ \bg_1=[\be_1,[\be_2,\be_3]]\ ,\ \bg_2=[\be_2,[\be_3,\be_1]]\;.\label{defcom}
\end{eqnarray}
This proves the surjectivity of $\phi$.\vspace{1mm}

 Finally, we must prove the injectivity of $\phi$.  Let us define $\bar \bof_1, \bar \bof_2, \bar \bof_3,\bar  \bg_1$ and $\bar \bg_2$ the elements generated from $\bar \be_1,\bar \be_2,\bar \be_3$ by analogy with (\ref{defcom}).
We must prove that $\{\bar \be_i,\bar \bof_i,\bar \bg_j\ | \ i=1,2,3,j=1,2\}$ satisfy 
relations \eqref{eq:AW21}-\eqref{eq:AW26} if they satisfy $\bar \be_1,\bar \be_2,\bar \be_3$ satisfy relations \eqref{eq:OA3}-\eqref{eq:OA3q}. Equations \eqref{eq:AW21} are satisfied because of the definition of 
$\bar \bof_i$. Equations \eqref{eq:AW22} for $i=j=1,2$ are satisfied due to the definition of $\bar \bg_1$ and $\bar \bg_2$. Then, one gets
\begin{equation}
 [\bar \be_3,\bar \bof_3]=[\bar \be_3,[\bar \be_1,\bar \be_2]]= [[\bar \be_3,\bar \be_1],\bar \be_2]+[\bar \be_1,[\bar \be_3,\bar \be_2]]=  -\bar \bg_1-\bar \bg_2
\end{equation}
using the Jacobi identity. This  proves eq. \eqref{eq:AW22} for $i=j=3$. 
By using relation \eqref{eq:OA3}, we show that
\begin{equation}
 [\bar \be_1,\bar \bof_2]=[\bar \be_1,[\bar \be_3,\bar \be_1]]=-\bar \be_3\ ,
\end{equation}
which proves eq. \eqref{eq:AW22} for $i=1$ and $j=2$. The remaining relations \eqref{eq:AW22} are proven similarly.
Now, using the Jacobi identity and relations \eqref{eq:OA3}-\eqref{eq:OA3q}, one gets
\begin{equation}
 [\bar\be_1,\bar\bg_2]=[\bar\be_1,[\bar\be_2,[\bar\be_3,\bar\be_1]]]
 =[[\bar\be_1,\bar\be_2],[\bar\be_3,\bar\be_1]]+[\bar\be_2,[\bar\be_1,[\bar\be_3,\bar\be_1]]]
 =-[\bar\be_2,\bar\be_3]+\alpha\bar \be_1-[\bar\be_2,\bar\be_3]=-2\bar \bof_1+\alpha\bar \be_1\ ,
\end{equation}
which proves relations \eqref{eq:AW23} for $i=1$ and $j=2$.
All the other relations \eqref{eq:AW23}-\eqref{eq:AW26} are proven in the same way which concludes the proof.
\end{proof}

\begin{rem}
From Proposition \ref{pro1} and the presentation \eqref{eq:OAn} for $N=3$, the $sl_3$-Askey-Wilson algebra ${AW}(a_{2}^{(1)})$ is the quotient of $sl_3$-Onsager algebra by \eqref{eq:OA3q}.
\end{rem}

\begin{rem}There exists one more equivalent presentation of the $sl_3$-Askey-Wilson algebra by replacing relations \eqref{eq:OA3q} by
\begin{eqnarray} 
[\bar \be_i ,[\bar \be_j,[\bar \be_k,\bar \be_i]]]=\alpha\epsilon_{ijk}\bar \be_i-2[\bar \be_j,\bar \be_k]\ \qquad \mbox{for} \quad 1\leq i \neq j\neq k \leq 3\ .
\end{eqnarray}
\end{rem}
\vspace{1mm}

\subsection{The $sl_4$-Askey-Wilson algebra $AW(a_{3}^{(1)})$}
Following the previous subsection, the simplest non-trivial solution of  (\ref{eq:Al}), (\ref{eq:Al2})  of the form (\ref{eq:Bx}) for $N=4$ is given in the definition below.
\begin{defn}\label{def:AW3}
The $sl_4$-Askey-Wilson algebra, denoted $AW(a_{3}^{(1)})$, is generated by $\{\be_i,\bof_i,\bg_i,\bh_j\ | \ i=1,2,3,4,j=1,2,3\}$ subject to the defining relations \eqref{eq:Al} with
\begin{equation}\label{eq:Baw3}
B(x)=\frac{2}{\alpha -x-x^{-1}}
\begin{pmatrix}
 \frac{3}{4}\bh_1+\frac{1}{2}\bh_2+\frac{1}{4}\bh_3 & \bg_1+x^{-1}\be_1 &
 \bof_1+  x^{-1}\bof_3 & -\be_4-x^{-1}\bg_4 \\
-\bg_1 -x\be_1  & -\frac{1}{4}\bh_1+\frac{1}{2}\bh_2+\frac{1}{4}\bh_3 & \bg_2+x^{-1}\be_2 &
\bof_2+x^{-1}   \bof_4 \\
 \bof_1+x \bof_3 &  -\bg_2-x\be_2 & -\frac{1}{4}\bh_1-\frac{1}{2}\bh_2+\frac{1}{4}\bh_3 &
   \bg_3+x^{-1}\be_3 \\
 \be_4+x\bg_4  & \bof_2+x \bof_4 & -\bg_3-x\be_3 &
   -\frac{1}{4}\bh_1-\frac{1}{2}\bh_2-\frac{3}{4}\bh_3 \\
\end{pmatrix}.
\end{equation}
\end{defn}
Inserting (\ref{eq:Baw3}) into the relations \eqref{eq:Al}, one gets equivalently the following defining relations of the $sl_4$-Askey-Wilson algebra:
\begin{eqnarray}
&&\left[\be_i,\be_{i+1}\right]=\bof_{i+2}\,,\label{eq:AW31}\\
&&\left[\be_i,\be_{i+2}\right]=0\,,\nonumber\\
&&\left[\be_i,\bof_{j}\right]=\delta_{i,j}\bg_{i+1}-\delta_{j,i-1}\bg_{i-1}-\delta_{j,i+1}\be_{i-1}+\delta_{j,i+2}\be_{i+1}\,,\label{eq:AW32}\\
&&\left[\be_i,\bg_{j}\right]=-\delta_{i,j}\bh_i+\delta_{j,i+1}\bof_i-\delta_{j,i-1}\bof_{i-1}\,,\label{eq:AW33}\\
&&\left[\be_i,\bh_{j}\right]=-\delta_{i,j}\left(2\alpha\be_i+4\bg_i\right)+\left(\delta_{j,i+1}+\delta_{j,i-1}\right)\left(\alpha\be_i+2\bg_i\right)\,,\label{eq:AW34}\\
&&\left[\bof_i,\bof_{i+1}\right]=0\,,\label{eq:AW35}\\
&&\left[\bof_i,\bof_{i+2}\right]=-\bh_{i}-\bh_{i+1}\,,\nonumber\\
&&\left[\bof_{i},\bg_j\right]=-\delta_{j,i-1}\left(\alpha\be_{i-2}+\bg_{i-2}\right)+\delta_{j,i-2}\left(\alpha\be_{i-1}+\bg_{i-1}\right)-\delta_{i,j}\be_{i+1}+\delta_{j,i+1}\be_{i}\,,\label{eq:AW36}\\
&&\left[\bof_{i},\bh_j\right]=\left(\delta_{i,j}-\delta_{j,i-1}+\delta_{j,i+1}-\delta_{j,i-2}\right)\left(\alpha\bof_{i}+2\bof_{i+2}\right)\,,
\label{eq:AW37}
\\
&&\left[\bg_i,\bg_{i+1}\right]=-\alpha\bof_i-\bof_{i+2}\,,\label{eq:AW38}\\
&&\left[\bg_i,\bg_{i+2}\right]=0\,,\nonumber\\
&&\left[\bg_i,\bh_{j}\right]=\delta_{i,j}\left(4\be_i+2\alpha\bg_i\right)-\left(\delta_{j,i+1}+\delta_{j,i-1}\right)\left(2\be_i+\alpha\bg_i\right)\,,\label{eq:AW39}\\
&&\left[\bh_i,\bh_{j}\right]=0\,,\label{eq:AW310}
\end{eqnarray}
where all indices should be taken $mod(4)$ and $\bh_4=-\bh_1-\bh_2-\bh_3$.\vspace{1mm}

Similarly to the case $N=3$, an alternative presentation for the  $sl_4$-Askey-Wilson algebra can be identified. The proof of the following proposition essentially follows the same steps as for the case $N=3$. For this reason, details are omitted.
\begin{prop}\label{pro2}
 The $sl_4$-Askey-Wilson algebra $AW(a_{3}^{(1)})$ is isomorphic to the algebra generated by $\{\bar \be_1,\bar \be_2,\bar \be_3,\bar \be_4\}$ and subject to the following relations:
\begin{eqnarray}
&&\big[\bar \be_i,\big[\bar \be_i,\bar \be_{i\pm1}\big]\big]=\bar \be_{i\pm1} \ , \label{eq:OA4a}\\
&&\big[\bar \be_i,\bar \be_{i+2}\big]=0\ , \label{eq:OA4b}\\
&&\Big[\big[\bar \be_i,\bar \be_{i+1}\big],\big[\bar \be_{i+1},\bar \be_{i+2}\big]\Big]=0 \ , \label{eq:OA4c}\\
&&\Big[
\big[\bar \be_i,\bar \be_{i+1}\big],
\big[\bar\be_{i+1},\big[\bar\be_{i+2},\bar\be_{i+3}\big]\big]
\Big]=-\alpha \bar\be_{i+1}-\big[\bar\be_i,\big[\bar\be_{i+2},\bar\be_{i+3}\big]\big] \ , \label{eq:OA4d}\\
&&\bigg[
\big[\bar\be_{i},\big[\bar\be_{i+1},\bar\be_{i+2}\big]\big],
\Big[\bar\be_{i+3},\big[\bar\be_{i},\big[\bar\be_{i+1},\bar\be_{i+2}\big]\big]\Big]
\bigg]=-4\bar\be_{i+3}+2\alpha \big[\bar\be_{i},\big[\bar\be_{i+1},\bar\be_{i+2}\big]\big] \ , \label{eq:OA4e}
 \end{eqnarray}
for $i=1,2,3,4$ and where all indices should be taken $mod(4)$.
\end{prop}

\subsection{Generalizations}
By analogy with the analysis of the previous subsections, we can propose a general form of $B(x)$ that will generate the $sl_N$-Askey-Wilson algebra, denoted $AW(a_{N-1}^{(1)})$. We find:
\beqa
B(x)= \frac{2}{(\alpha +(-1)^{N+1} x-x^{-1})}\left(\bold{b}_{i,j}(x)\right)_{1\leq i,j\leq N}
\eeqa
where
\ben\label{eq:BawN-1}
\qquad &&\bold{b}_{i,i}(x)=\sum_{l=1}^{N-1}c_{i}^l\bof_{l,l-1}   \qquad \mbox{with}\quad c_{i}^l = \frac{N-l}{N} \quad \mbox{if} \quad i\leq l, \quad  c_{i}^l = -\frac{l}{N} \quad \mbox{if} \quad i>l,  \quad l=1,...,N-1 \,,\nonumber\\
&&\bold{b}_{i,j}(x)=(-1)^N\left(-\bof_{i+1,i-1}+x^{-1}\be_i\right)\quad\mbox{for}\quad j-i=1,\, i<j \ ,\nonumber\\
&&\bold{b}_{i,j}(x)=(-1)^{N+1}\be_N+x^{-1}\bof_{1,N-1}\quad\mbox{for}\quad (i,j)=(1,N) \ ,\nonumber\\
&&\bold{b}_{i,j}(x)= (-1)^N\bof_{j+1,j-1}- x\be_j\quad\mbox{for}\quad i-j=1,\, i>j\ ,\nonumber\\
&&\bold{b}_{i,j}(x)=\be_N-x\bof_{1,N-1}\quad\mbox{for}\quad (i,j)=(N,1) \ ,\nonumber\\
&&\bold{b}_{i,j}(x)=(-1)^{(j-i)(N+1)}\left(
\bof_{j,i-1}+(-1)^{j-i}x^{-1}
\bof_{i,j-1}
\right)
\quad\mbox{for}\quad i<j, \,i,j\,\,\,\mbox{not~adjacent}\,,\nonumber\\
&&\bold{b}_{i,j}(x)=
(-1)^{(i-j)N}(\bof_{i,j-1}+(-1)^{(i-j)+N}x
\bof_{j,i-1})
\quad\mbox{for}\quad i>j, \,i,j\,\,\,\mbox{not~adjacent}\,\nonumber \ ,
\een
and $\bof_{i,j}=[\be_i,[\be_{i+1},[\dots,[\be_{j-1},\be_{j}] \dots]]]$ is a multiple commutator.
Note that indices in the multiple commutators should be taken modulo $N$.

\section{Concluding remarks}
In this letter, in a first part we have obtained a FRT presentation for the $sl_N$-Onsager algebra from which the corresponding current algebra and a generating function for mutually commuting elements have been derived. Thus, besides the two original presentations of the $sl_N$-Onsager algebra  given in \cite{UI}, this third presentation connects with the realm of Yang-Baxter algebras and related techniques such as the Bethe ansatz, opening new perspectives in the context of quantum integrable systems. In a second part, we have shown that quotients of the  $sl_N$-Onsager algebra lead to higher rank extensions of the classical analog ($q=1$) of the Askey-Wilson algebra \cite{Z91,Z92} following the strategy recently used in \cite{BC} based on the FRT presentation of the Onsager algebra \cite{Ons44}. \vspace{1mm}

Although the results are not reported here, let us mention that a similar analysis can be pursued if one considers the second automorphism $\theta_2$ given in (\ref{eq:aut2}) instead of $\theta_1$. In this case, we have obtained a different type of $sl_N$-Onsager algebra. The corresponding current algebra and mutually commuting elements have been derived. For $N=2$, this algebra was introduced and studied in \cite{BC13} (see also \cite{BBC}); it was called the augmented Onsager algebra, as it is a special case of the augmented trididiagonal algebra introduced by Ito and Terwilliger in \cite{IT10}. Thus, for $N>2$, we have obtained a $sl_N$-augmented Onsager algebra. 
\vspace{1mm}

Several open problems may be now considered.\vspace{1mm}

Firstly, in order to complete the picture,  one could extend the analysis  here presented - that is based on the framework of \cite{BBC} specialized here to the case of $a_{N-1}^{(1)}$ - to the case of any affine Lie algebra. The corresponding FRT presentation would provide a complete description of the generalized Onsager algebras $\cO(\widehat{g})$ associated with any affine Lie algebra $\widehat{g}$ \cite{UI,DU04,S18}. A related issue is the construction of a current presentation for the $q-$deformed analog of the generalized Onsager algebras $\cO_q(\widehat{g})$ \cite{BB1,BF} that is still an open problem. The solution would find applications in the analysis of the thermodynamic limit of open spin chains with higher rank symmetries, see e.g. \cite{FK}. By analogy with the approach presented here, following \cite{WZ,BK2} the defining relations of the higher rank Askey-Wilson algebras $AW(\widehat{g})$ for $q\neq 1$ could also be extracted from the reflection equation algebra associated with the R-matrix of $U_q(\widehat{g})$. 
 \vspace{1mm}

Secondly, as mentionned in the introduction, in the context of the theory of special functions it is known that the Askey-Wilson algebra provides an algebraic scheme for the classification of all orthogonal polynomials in one variable  that enjoy a bispectral property \cite{Z91,Z92}. By analogy,  an interesting problem would be to characterize the family of multivariate special functions associated with the higher rank classical Askey-Wilson $AW(a_{N-1}^{(1)})$  (and its $q-$analog) and to understand the notion of  bispectrality in this context.\vspace{1mm}

\vspace{0.5cm}
\noindent{\bf Acknowledgements:} 
 R.P. thanks the Institut Denis Poisson for hospitality where part of this work has been done. R.P. is supported by
the S\~ao Paulo Research Foundation (FAPESP) and by the Coordination for the Improvement of Higher Education Personnel (CAPES), grants \# 2017/02987-8 and \#88881.171877/2018-01.
P.B.  and N.C. are supported by C.N.R.S.

\end{document}